\documentclass[11pt,english]{paper}
\usepackage[T1]{fontenc}
\usepackage[latin9]{inputenc}
\usepackage{float}
\usepackage{amsthm}
\usepackage{amsmath}
\usepackage{amssymb}
\usepackage{graphicx}

\makeatletter
\newcommand{\lyxaddress}[1]{
\par {\raggedright #1
\vspace{1.4em}
\noindent\par}
}
  \theoremstyle{plain}
  \newtheorem*{thm*}{\protect\theoremname}
\theoremstyle{plain}
\newtheorem{thm}{\protect\theoremname}
  \theoremstyle{plain}
  \newtheorem{cor}[thm]{\protect\corollaryname}

\makeatother

\usepackage{babel}
  \providecommand{\corollaryname}{Corollary}
  \providecommand{\theoremname}{Theorem}
\providecommand{\theoremname}{Theorem}

\begin{document}

\title{Quantum Chaos and Quantum Computing Structures}

\author{Carlos Pedro Gonçalves}

\institution{{\small Instituto Superior de Ciências Sociais e Políticas (ISCSP)
- Technical University of Lisbon}}

\maketitle

\lyxaddress{cgoncalves@iscsp.utl.pt}
\begin{abstract}
A system of quantum computing structures is introduced and proven
capable of making emerge, on average, the orbits of classical bounded
nonlinear maps on $\mathbb{C}$ through the iterative action of path-dependent
quantum gates. The effects of emerging nonlinear dynamics and chaos
upon the quantum averages of relevant observables and quantum probabilities
are exemplified for a version of Chirikov's standard map on $\mathbb{C}$.
Both the individual orbits and ensemble properties are addressed so
that the Poincaré map for Chirikov's standard map, in the current
quantum setting, is reinterpreted in terms of a quantum ensemble which
is then formally introduced within the formalized system of quantum
computing structures, in terms of quantum register machines, revealing
three phases of quantum ensemble dynamics: the regular, the chaotic
and an intermediate phase called \emph{complex quantum stochastic
phase} which shares similarities to the \emph{edge of chaos} notion
from classical cellular automata and classical random boolean networks'
evolutionary computation.\end{abstract}
\begin{keywords}
Quantum Computation, Quantum Chaos, Emergence, Quantum Complex Systems
Science.
\end{keywords}

\section{Quantum Chaos and Complex Systems Science}

Complex quantum systems science, primarily concerned with complex
quantum systems \cite{key-6,key-9,key-10,key-16,key-17,key-18}, has
grown in interdisciplinary applications to include quantum econophysics
\cite{key-14,key-15,key-19,key-20}, showing effective results in
dealing with financial market turbulence and financial risk modelling
\cite{key-7,key-8}. Such results open up the way for the growth of
a \emph{quantum complex systems research} program%
\footnote{A point established, in particular, by Saptsin and Soloviev in their
articles on conceptual foundations for quantum econophysics and its
relevance for complex systems science \cite{key-19,key-20}.%
}, which includes both complex quantum systems as well as other systems
to which the quantum formalism, mathematically generalized from quantum
computation theory, provides solutions that the classical formalism
does not.

Given this recent growth, we can now speak of a \emph{quantum complex
systems science} (on purpose reversion of the first two words), that
is, a branch of the complexity sciences dealing with both \emph{complex
quantum systems} and the generalization of the quantum formalism and
methods to other areas of research dealing with \emph{complex systems}
in general. The current work is a contribution to such a research
program, addressing emerging chaotic dynamics in \emph{qunat}-based
quantum computation \cite{key-23}.

By quantum chaos theory it is understood, here, the formal combination
of quantum theory and chaos theory. Two lines of research have developed
within quantum chaos theory: the first one has addressed quantum counterparts
of classically chaotic systems \cite{key-21}, the second one has
addressed chaotic quantization \cite{key-1}. In the current work,
we prove that \emph{qunat}-based quantum computation can make emerge
classical nonlinear dynamics and, in particular, chaotic orbits%
\footnote{We are applying here Varela's notion of \emph{enaction} \cite{key-22},
thus, when we state \emph{to make emerge} we are using this authors'
concept of \emph{to enact} as \emph{to make emerge}, so that a \emph{quantum
computing structure} can \emph{make emerge}, from that computation,
an on average classically chaotic dynamics.%
} (\emph{section 2.}). The present work is, therefore, developed from
the combination of the first line of research on quantum chaos with
quantum computation, showing that classically chaotic orbits can emerge
from quantum computation.

In \emph{section 3.}, we address the example of a version of Chirikov's
standard map \cite{key-4} on $\mathbb{C}$ that exemplifies chaotic
dynamics emerging from quantum dynamics and the consequences for the
underlying quantum system in regards to the quantum state dynamics
and quantum averages.

In \emph{section 4.}, we address the ensemble representation and how,
for a system of quantum registers, the Poincaré maps can be used as
a visualization tool for emerging nonlinear dynamics in a system of
$m$ quantum computing register machines. Three dynamical phases are
identified in the ensemble quantum dynamics one of them, called \emph{complex
quantum stochastic phase}, is addressed in greater detail and linked
to classical complex systems science's notion of \emph{edge of chaos}.

In \emph{section 5.}, the main consequences of the present work for
quantum mechanics and quantum complex systems science are addressed.

\section{Quantum Computing Structures, Chaos and Nonlinear Dynamics}

Let a \emph{qunat bosonic} quantum computing structure be defined
as a triple $\left(\mathcal{H},\mathcal{A},\mathcal{U}\right)$, where
$\mathcal{A}$ is an operator structure comprised of the bosonic creation
and anihilation operators $a^{\dagger}$ and $a$, respectively, along
with the commutator relations $\left[a,a^{\dagger}\right]=1$, $\mathcal{H}$
is the Hilbert space spanned by the basis of the number operator's
eigenstates $\hat{N}\left|n\right\rangle =n\left|n\right\rangle $,
for $n=0,1,2,...$, $\mathcal{U}$ is the family of \emph{qunat} quantum
gates defined as the unitary operators on $\mathcal{H}$.

Now let $\mathcal{C}\subset\mathcal{H}$ be the collection of coherent
states defined as the eigenstates of the anihilation operator, that
is:
\begin{equation}
\mathcal{C}:=\left\{ a\left|c\right\rangle =c\left|c\right\rangle :c\in\mathbb{C}\right\} 
\end{equation}
Define $Coh$ as the mathematical category \cite{key-13} whose object
set is $\mathcal{C}$ and morphisms are the unitary operators in $\hom(Coh)\subset\mathcal{U}$,
satisfying the condition:
\begin{equation}
\hom(Coh):=\left\{ U(c',c)\in\mathcal{U}:U(c',c)\left|c\right\rangle =\left|c'\right\rangle \right\} 
\end{equation}
that is, each operator in $\mathcal{U}_{Coh}$ connects two, and only
two, coherent states, unitarily transforming one into the other. Given
the definition of $\hom(Coh)$ it immediately follows that any $U(c',c)$
decomposes as the product of two optical displacement operators \cite{key-6}:
$U(c',c)=D(c')D(c)^{\dagger}$ (with $D(h):=\exp\left(ha^{\dagger}-h^{*}a\right)$).
Defining composition operation as the product between the two unitary
operators $U\left(c'',c'\right)\circ U\left(c',c\right)=U\left(c'',c'\right)U\left(c',c\right)$
it immediately follows, from the relation to the optical displacement
operators, that the closure of $\hom(Coh)$ with respect to compostion
and the associativity can be proven to be satisfied by $\hom(Coh)$.
Defining identity morphic connection by $1_{c}=U(c,c)=D(c)D(c)^{\dagger}$,
it follows that all of the properties of a category are satisfied
by $Coh$, which can be proved from the optical displacement operator
decomposition.

Now, a \emph{quantum computing basis} of $Coh$ (not to be confused
with the vector basis of the Hilbert space in $\left(\mathcal{H},\mathcal{A},\mathcal{U}\right)$)
is defined as a subcategory $S\subseteq Coh$. Given a \emph{quantum
computing basis} $S$ of $Coh$ we state that a complex-valued function
$F:X\mapsto Y$, with $X,\: Y\subseteq\mathbb{C}$, emerges from\emph{
$S$} if, and only if, the following three conditions are met: 
\begin{enumerate}
\item $X_{Q}\cup Y_{Q}\subseteq S$, for $X_{Q}:=\left\{ \left|c\right\rangle \in\mathcal{C}:c\in X\right\} $,
$Y_{Q}:=\left\{ \left|c\right\rangle \in\mathcal{C}:c\in Y\right\} $,
that is, the collection of coherent states corresponding to the elements
of the domain of $F$ and the collection of coherent states corresponding
to the elements of the codomain of $F$ are both in $S$;
\item There is a $\hom\left(S_{F}\right)\subseteq\hom\left(S\right)$ such
that, for every $\left|c\right\rangle \in X_{Q}$, there is one, and
only one, $U\left(c',c\right)\in\hom\left(S_{F}\right)$ satisfying
$U(c',c)\left|c\right\rangle \in Y_{Q}$ and it is such that $c'=F(c)$,
that is, we can find a subcollection of morphisms which maps each
coherent state $\left|c\right\rangle $ in $X_{Q}$ to one, and only
one, coherent state in $Y_{Q}$, and this coherent state is the state
corresponding to the complex number that is the image of $c$ in $Y$
under $F$.
\end{enumerate}
We call the structure $Q_{F}:=\left(X_{Q}\cup Y_{Q},\hom\left(S_{F}\right)\right)$
the \emph{quantum generator} of $F$ in $S$. The generator's quantum
computational pattern \emph{enacts} a mathematical function through
the unitary gate actions upon $X_{Q}$. Varela's notion of \emph{enaction}
is the systemically proper notion to use in this case \cite{key-22},
quantum computation is about a systemic concrete, it takes place with
systemic effects in the system's quantum state, and, in this case,
the function is not computed as an direct output of a quantum gate,
it emerges out of the pattern of quantum gates and computations performed
by a computing system, which leads to the notion of emergence of a
function out of the quantum computation performed by a \emph{quantum
computing basis}, by taking the average of the anihilation operator,
that is, we have the following matching between (a concrete) quantum
computation%
\footnote{Concrete since it refers to a systemic (physical) concrete that is
computing (and without computing system there is no computation),
a fact that is independent of the interpretation of quantum mechanics.%
} and (an abstract) complex-valued function%
\footnote{The function belongs to the level of the mathematical language which
addresses, in this case, the eigenvalues of the anihilation operator
in terms of an abstract complex number system.%
}:
\begin{equation}
\begin{array}{ccc}
\left|c\right\rangle \in X_{Q} & \overset{U(c',c)}{\longrightarrow} & \left|c'\right\rangle \in Y_{Q}\\
\downarrow &  & \downarrow\\
\left\langle c|a|c\right\rangle =c\in X & \underset{F}{\longrightarrow} & \left\langle c'|a|c'\right\rangle =F(c)\in Y
\end{array}
\end{equation}

Thus, given the input $\left|c\right\rangle $ the quantum generator
$Q_{F}$ unitarily transforms it in the output $\left|c'\right\rangle =\left|F(c)\right\rangle $,
this means that, evaluating, at the input, the quantum average of
the anihiliation operator, we obtain $\left\langle c|a|c\right\rangle $
which coincides with $c$ because we are dealing with coherent states,
on the other hand, and since we are dealing with coherent states,
we obtain $\left\langle c'|a|c'\right\rangle =c'$, the emergence
of $F$ from the quantum computation implies, then, that $c'=F(c)$.

In systemic terms, we are dealing with an emergence because the function
is not the directing agent, indeed, this is not a \emph{top-down}
process, the quantum computation leads to a relation that can be synthesized
by the function as identifiable in the relation between input $\left\langle c|a|c\right\rangle $
and output $\left\langle c'|a|c'\right\rangle $, for each $c$ in
the domain of $F$, thus, the complex-valued function $F$ emerges
from the quantum computation.
\begin{thm*}
(Quantum Emergence of Complex-Valued Functions) Any complex-valued
function $F:X\mapsto Y$, with $X,\: Y\subseteq\mathbb{C}$, can emerge
from a quantum computing basis of $Coh$.\end{thm*}
\begin{proof}
From conditions (1) and (2), and from the definition of $Coh$, it
follows that there is a one-to-one and onto correspondence between
$\mathcal{C}$ (the collection of objects of $Coh$) and the complex
numbers $\mathbb{C}$, this correspondence, defined as $\tau:\mathcal{C}\mapsto\mathbb{C}$,
can be built through an operation of taking the quantum average of
the anihilation operator, which leads to a corresponding complex number
corresponding to the coherent state, that is:
\begin{equation}
\tau\left(\left|c\right\rangle \right):=\left\langle c\left|a\right|c\right\rangle =c
\end{equation}

Now, any complex-valued function $F:X\mapsto Y$, with $X$, $Y\subseteq\mathbb{C}$
is such that, to each element of the domain $c\in X$, it associates
one, and only one element in the codomain $Y$, let, then, $S$ be
a subcategory of $Coh$ satisfying $\tau^{-1}(X)\cup\tau^{-1}(Y)\subseteq S$,
we know that such a subcategory exists from the correspondence between
$\mathcal{C}$ and $\mathbb{C}$. Let, also, the subcollection of
morphisms $\hom\left(\tau^{-1}(X)\cup\tau^{-1}(Y)\right)\subseteq\hom\left(S\right)$
be such that for each $\left|c\right\rangle \in\tau^{-1}(X)$ there
is one, and only one $U\left(c',c\right)$ for $c'\in\tau^{-1}(Y)$
and $c'$ is such that $c'=F(c)$, thus, the output in $\tau^{-1}(Y)$
of each object (coherent state) in $\tau^{-1}(X)$, under the quantum
gates in $\hom\left(\tau^{-1}(X)\cup\tau^{-1}(Y)\right)$ matches
the function $F$, this matching is in accordance with the following
square: 
\begin{equation}
\begin{array}{ccc}
\left|c\right\rangle  & \overset{U(c',c)}{\longrightarrow} & \left|c'\right\rangle \\
\tau\downarrow &  & \downarrow\tau\\
c & \underset{F}{\longrightarrow} & c'
\end{array}
\end{equation}
so that $F$ is locally isomorphic to each quantum computation under
the structure $Q_{F}:=\left(\tau^{-1}(X)\cup\tau^{-1}(Y),\hom\left(\tau^{-1}(X)\cup\tau^{-1}(Y)\right)\right)$
which is, therefore, a \emph{quantum generator} of $F$. We know that
such a subcategory $S$ with a structure like $Q_{F}$ exists, since
the object structure $\tau^{-1}(X)\cup\tau^{-1}(Y)$ exists in $Coh$,
which follows directly from the correspondence between $\mathcal{C}$
and $\mathbb{C}$, and given that, for any pair of coherent states
$\left|c\right\rangle $ and $\left|c'\right\rangle $, there is a
corresponding \emph{quantum gate} in $\hom\left(Coh\right)$ such
that $U(c',c)\left|c\right\rangle =\left|c'\right\rangle $, then,
it follows that the morphisms that compose $\hom\left(\tau^{-1}(X)\cup\tau^{-1}(Y)\right)$,
with the structure laid out above, also exist, therefore, we find
that any complex-valued function can indeed emerge from a quantum
computing basis of $Coh$.
\end{proof}
This theorem has a relevant computational consequence, in the following
sense: while the quantum computing structure $\left(\mathcal{H},\mathcal{A},\mathcal{U}\right)$
works with a discrete (number operator) basis with corresponding eigenvalues
expressed in natural numbers, it is capable of performing a computational
jump to hypercomputation, being able to quantum computationally make
emerge a function on a mathematical continuum within a quantized setting
of basis states.

The immediate corollary of this theorem is that:
\begin{cor}
Any bounded nonlinear map on an complex interval $\mathbb{I}\subset\mathbb{C}$,
$F:\mathbb{I}\mapsto\mathbb{I}$ can emerge from a quantum computing
basis of $Coh$.
\end{cor}
Not only can a \emph{bounded nonlinear map} emerge from a \emph{quantum
computing basis} of $Coh$, but also its iterates:
\begin{cor}
Given any bounded nonlinear map on an complex interval $\mathbb{I}\subset\mathbb{C}$,
$F:\mathbb{I}\mapsto\mathbb{I}$, each of its iterates $F^{s}$, with
$s=1,2,...,$ being complex-valued functions can also emerge from
a quantum computing basis of $Coh$.
\end{cor}
From these results it, then, follows that classical nonlinear maps'
dynamics can also emerge from a quantum computing basis of $Coh$,
a statement that can be placed in the form of a theorem:
\begin{thm*}
(Emergent Classical Orbits) Let $F$ be a bounded complex-valued nonlinear
map on an interval $\mathbb{I}\subset\mathbb{C}$, then, each orbit
of $F$ in $\mathbb{I}$ can emerge from a quantum computing basis
of $Coh$.\end{thm*}
\begin{proof}
Let $F$ be a bounded complex-valued nonlinear map on an interval
$\mathbb{I}\subset\mathbb{C}$, From the \emph{quantum emergence of
complex-valued functions theorem} and corollary 1. we know that there
is a quantum computing basis of $Coh$ from which this map can emerge,
denote such a basis by $S$, by replacing $X$ and $Y$ by $\mathbb{I}$,
in the the \emph{quantum emergence of complex-valued functions theorem}'s
proof, we arrive at the resulting quantum generator of $F$ in $S$:
$\left(\tau^{-1}\left(\mathbb{I}\right),\textrm{hom}\left(\tau^{-1}\left(\mathbb{I}\right)\right)\right)$.

Now, let an orbit of $F$ in $\mathbb{I}$ be defined by a sequence
$\left\{ c_{s}\right\} _{s}$, $s=0,1,2,...$, with the following
iterative scheme:
\begin{equation}
c_{s}=F\left(c_{s-1}\right)
\end{equation}
then, from the structure of $\left(\tau^{-1}\left(\mathbb{I}\right),\textrm{hom}\left(\tau^{-1}\left(\mathbb{I}\right)\right)\right)$,
we know that, for each such orbit, there is a corresponding sequence
of coherent states $\left\{ \left|c_{s}\right\rangle \right\} _{s}$,
for $s=0,1,2,...$, obeying the scheme:
\begin{equation}
\left|c_{s}\right\rangle =U\left(F\left(c_{s-1}\right),c_{s-1}\right)\left|c_{s-1}\right\rangle 
\end{equation}
Each orbit of $F$ can, then, be addressed as the result from taking
the average of the anihilation operator for a sequence of coherent
states which, in turn, can be placed in isomorphic correspondence
with the classical orbit under $\tau\left(\left|c_{s}\right\rangle \right)=\left\langle c_{s}\right|a\left|c_{s}\right\rangle =c_{s}$,
that is:
\begin{equation}
\left\{ \left|c_{s}\right\rangle \right\} _{s}\overset{\tau}{\longrightarrow}\left\{ c_{s}\right\} _{s}
\end{equation}
with the local connections: 
\begin{equation}
\begin{array}{ccc}
\left|c_{s-1}\right\rangle  & \overset{U\left(c_{s},c_{s-1}\right)}{\longrightarrow} & \left|c_{s}\right\rangle \\
\tau\downarrow &  & \downarrow\tau\\
c_{s-1} & \underset{F}{\longrightarrow} & c_{s}
\end{array}
\end{equation}
which allows us to conclude that each orbit of $F$ can emerge from
a \emph{quantum computing basis} of $Coh$, emerging from the coherent
state unitary tansition in the form of a sequence of quantum averages
of the anihilation operator (as per definition of $\tau$).
\end{proof}
In the last theorem, it follows that we are dealing with a form of
path-dependent quantum computation, in the sense that each classical
orbit emerges from a specific sequence of unitary operators actions
on an initial condition $\left|c_{0}\right\rangle $, a sequence that
matches the iterations of the classical map from the corresponding
initial condition $\tau\left(\left|c_{0}\right\rangle \right)=c_{0}$.

If a bounded complex-valued nonlinear map on an interval $\mathbb{I\subset\mathbb{C}}$
is chaotic on $\mathbb{J\subseteq\mathbb{I}},$ then, under this last
theorem, it follows that the chaotic dynamics can emerge from the
quantum dynamics.

To each chaotic orbit there corresponds a specific sequence of coherent
states such that the chaotic orbit emerges from the coherent state
unitary transition in the form of a sequence of quantum averages of
the anihilation operator. In turn, the sequence of Poisson number
distribution of coherent states is such that the probability for each
number state becomes expressed as a nonlinear function of a chaotic
orbit:
\begin{equation}
P_{s}(n)=\left|\left\langle n|c_{s}\right\rangle \right|^{2}=\frac{\left|c_{s}\right|^{2n}}{n!}\exp\left(-\left|c_{s}\right|^{2}\right)
\end{equation}
the quantum probabilities will thus also show chaotic fluctuations
as the quantum computation proceeds. Furthermore, for initial coherent
states with $\left|c_{0}\right\rangle $ and $\left|c_{0}'\right\rangle $,
whose corresponding anihilation operator eigenvalues are in a small
neighborhood of each other, the resulting state sequences will diverge,
with respect to their corresponding anihilation operator's eigenvalues,
in accordance with the sensitive dependence upon initial conditions
that characterizes classical chaotic dynamics, this will also leave
its mark upon the probabilities, which will also differ after a few
quantum computations.

The quantum average of the number operator will also be a nonlinear
function of a chaotic orbit, indeed, we have:
\begin{equation}
\left\langle \hat{N}\right\rangle =\left|c_{s}\right|^{2}=\left|F\left(c_{s-1}\right)\right|^{2}
\end{equation}

We now illustrate this type of quantum computationally emergent chaos
and its consequences to the sequence of number operator averages,
for a version of Chirikov's standard map on $\mathbb{C}$.

\section{Chaos and Quantum Computation - Chirikov's Standard Map on $\mathbb{C}$ }

Let $S$ be a subcategory of $Coh$, and $\left\{ U\left(c_{s},c_{s-1}\right)\right\} _{s}$,
for $s=1,2,...,$ be a sequence of \emph{qunat gates} in $\hom\left(S\right)$,
such that:
\begin{equation}
\left|c_{s}\right\rangle =U\left(F\left(c_{s-1}\right),c_{s-1}\right)\left|c_{s-1}\right\rangle 
\end{equation}
and, for the polar coordinate representation $c_{s-1}=r_{s-1}e^{i\phi_{s-1}}$,
$F$ has the form:

\begin{equation}
F\left(r_{s-1}e^{i\phi_{s-1}}\right)=r_{s}e^{i\phi_{s}}
\end{equation}
\begin{equation}
r_{s}=r_{s-1}+K\sin\phi_{s-1}
\end{equation}

\begin{equation}
\phi_{s}=\phi_{s-1}+r_{s}
\end{equation}
both equations (14) and (15) are defined modulo $2\pi$.

Given this structure, with an initially coherent state $\left|r_{0}e^{i\phi_{0}}\right\rangle $,
with $r_{0}$ and $\phi_{0}$ both ranging between $0$ and $2\pi$,
the sequence of \emph{qunat gates} will act on the initial input leading
to a sequence of state transitions that leads to a sequence of quantum
averages of the anihilation operator coinciding with an orbit of a
complex-valued nonlinear map with the same functional shape as Chirikov's
standard map, introduced by Chirikov in the study of the kicked rotator
\cite{key-4}. There are quantum mechanical consequences of this dynamics
for observables of interest, in particular, in what regards the quantum
average of the number operator. Taking the expected value of the number
operator the following temporal dependence is found:

\begin{equation}
\left\langle \hat{N}\right\rangle _{s+1}=\left|c_{s+1}\right|^{2}=r_{s+1}^{2}=r_{s}^{2}+2r_{s}K\sin\phi_{s}+K^{2}\sin^{2}\phi_{s}
\end{equation}

An example is shown in Fig.1, of a transition to chaotic dynamics
for two different values of $K$ in terms of the power spectra for
the orbits of the quantum average of the number operator for a fixed
inicial condition. In the regular dynamics (Fig.1(left)) we can see
several significant frequencies, marking an oscillatory behavior in
the dynamics of $\left\langle \hat{N}\right\rangle _{s}$. As the
parameter $K$ is increased (Fig.1(right)) the dynamics of $\left\langle \hat{N}\right\rangle _{s}$
shows the presence of stochastic behavior with a few frequency regions
standing out. These frequency regions are more visible for $K$ near
the critical parameter $K_{c}=0.971635406...$ above which the last
KAM invariant torus is destroyed, for $K=0.971635406$, which is close
to 9 decimal places of the critical parameter, the presence of a broadband
spectrum of a stochastic signal shows not only a few frequency windows
that stand out in the spectrum but, also, a few frequency spikes (Fig.2).

Chirikov \cite{key-4} addressed such random motion in terms of the
nonlinear dynamics of interactions of resonances, more specifically,
the simultaneous effect on a nonlinear oscillator of several perturbations
with different frequencies such that in the limiting case of very
large overlapping resonant zones the system will show random motion,
a process that, in the present case, due to the nonlinear path dependence
of the quantum computation takes place in the sequence of phases and
amplitudes of the anihilation operator's eigenvalue orbit, leaving
its markings in the number operator's quantum averages' sequence.

Besides the single orbit analysis, as we move on to consider the dynamics
for different initial conditions, employing the Poincaré maps as visualization
tools we reach, in the quantum setting, a quantum statistical description
in terms of ensembles and start dealing with a complex quantum system
undergoing \emph{path-dependent quantum computation} making emerge
different dynamical regimes with conceptual implications for the study
of complex quantum systems. We, now, address this ensemble dynamics
in terms of a (bosonic) quantum register machine computational structure.

\section{Quantum Register Machine and Complex Quantum Systems' Dynamics}

The previous section's results allow us to obtain a general picture
of how dynamical stochasticity can emerge in the sequences of quantum
averages in a quantum computation. The quantum setting, on the other
hand, introduces a new perspective on the Poincaré map analysis of
the classical standard map in terms of quantum statistical ensembles,
which we now address.

If one looks at the Poincaré maps of Fig.3, within the quantum setting,
one may see that each of these is actually depicting the anihilation
operator eigenvalue orbit for an ensemble of 120 bosonic oscillators%
\footnote{The number was chosen purely out of representational purposes, since
120 is a high enough number, though not too high so that it would
make the whole plot thick black.%
}, each undergoing a local unitary evolution, this type of parallel
quantum computation, generalized to $m$-bosonic oscillators can be
addressed by a $m$-entry quantum register machine, formalized from
the single \emph{qunat bosonic} quantum computing structures.

Indeed, given the triple $\left(\mathcal{H},\mathcal{A},\mathcal{U}\right)$,
introduced in section 1., then, define the $m$-entry quantum register
machine by the tensor product structure $\left(\mathcal{H}^{m},\mathcal{A}^{m},\mathcal{U}^{m}\right)=\bigotimes_{j=1}^{m}\left(\mathcal{H}_{j},\mathcal{A}_{j},\mathcal{U}_{j}\right)$,
such that:
\begin{itemize}
\item Each $\left(\mathcal{H}_{j},\mathcal{A}_{j},\mathcal{U}_{j}\right):=\left(\mathcal{H},\mathcal{A},\mathcal{U}\right)$,
that is, each triple in the tensor product is a labeled copy of a
\emph{qunat bosonic} quantum computing structure;
\item $\mathcal{H}^{m}:=\bigotimes_{j=1}^{m}\mathcal{H}_{j}$ is a tensor-product
of $m$ copies of the same Hilbert space;
\item $\mathcal{A}^{m}:=\left\{ a_{i}:i=1,2,...,m\right\} \cup\left\{ a_{j}^{\dagger}:j=1,2,...,m\right\} $,
where the creation and anihilation operators obey the bosonic commutation
relations $\left[a_{i},a_{j}\right]=0$, $\left[a_{i}^{\dagger},a_{j}^{\dagger}\right]=0$
and $\left[a_{i},a_{j}^{\dagger}\right]=\delta_{i,j}$;
\item $\mathcal{U}^{m}$ is the family of unitary operators on $\mathcal{H}^{m}$.
\end{itemize}
In this case, the creation and anihilation operators can act on the
corresponding bosonic oscillator without affecting the rest, which
is in accordance with the fact that we are dealing with ensembles
that have a field-like computational behavior.

Given the above assumptions it follows that the number operators commute
$\left[\hat{N}_{i},\hat{N}_{j}\right]=0$ for $i\neq j$, and the
basis for the Hilbert space $\mathcal{H}^{m}$ is given by $\left\{ \left|n_{1},n_{2},...,n_{m}\right\rangle :n_{j}\in\mathbb{N}_{0},j=1,2,...,m\right\} $.
The registers are assumed to be sequential, much like it takes place
in a Turing machine, this means that the matching of a quantum ensemble
of bosonic oscillators to a register machine is unique up to a permutation
of registers, the fact that one works with a particular register sequence,
means that any results regarding the quantum ensemble, specifically
obtained for a quantum register machine in which the labeling is irrelevant,
can be generalized for all of the other permutations of registers.

Now, given a configuration of registers we can define the $\mathcal{C}_{m}\subset\mathcal{H}_{m}$
to be the collection of $m$-coherent states on $\mathcal{H}_{m}$
defined as:
\begin{equation}
\mathcal{C}_{m}:=\left\{ \left|c(1),c(2),...,c(m)\right\rangle :c(j)\in\mathbb{C},\: j=1,2,...,m\right\} 
\end{equation}
that is, each register in $\mathcal{C}_{m}$ corresponds to a bosonic
oscillator in a coherent state. We can, thus, define the $m$-coherent
states category $Coh^{m}$ as the product category of $m$ copies
of $Coh$ and such that each projection functor is labeled by the
corresponding register, that is:
\begin{equation}
\begin{array}{ccccc}
Coh_{i} & \overset{i}{\longleftarrow} & Coh^{m} & \overset{j}{\longrightarrow} & Coh_{j}\end{array}
\end{equation}
\begin{equation}
Coh_{i}\overset{id_{Coh_{i}}}{\longrightarrow}Coh_{j}
\end{equation}
where $id_{Coh_{i}}$ is the identify functor of $Coh_{i}$. From
the (18) and (19), it follows that the product category is also indexing
each copy of $Coh$ to a register, so that we have:
\begin{equation}
\left|c(1),c(2),...,c(m)\right\rangle \overset{j}{\longrightarrow}\left|c(j)\right\rangle 
\end{equation}

Since $Coh^{m}$ is a product category, the morphisms of $Coh^{m}$
must conserve the local morphic connections of each component category,
that is:
\begin{equation}
U\left(c'(j),c(j)\right)\left|...,c(j),...\right\rangle =\left|...,c'(j),...\right\rangle 
\end{equation}
which means that the unitary change of the $j$-th coherent state
to another coherent state in nothing alters the other registers' configuration.

For the Chirikov map, we, then, have the following quantum iterative
transition rule:
\begin{equation}
\prod_{j=1}^{m}U\left(F_{Chirikov}\left(c_{s}(j)\right),c_{s}(j)\right)\left|c_{s}(1),c_{s}(2),...,c_{s}(m)\right\rangle 
\end{equation}

The Poincaré map for Chirikov's standard map becomes, in this case,
a visual representation of this dynamics of each of the ensemble elements
with respect to the corresponding anihilation operator eigenvalue
in polar representation, thus, for chaotic orbits, even if $\left|c_{0}(1),c_{0}(2),...,c_{0}(m)\right\rangle $
begins with coherent states very close to each other, in terms of
the corresponding anihilation operators' eigenvalues, the sensitive
dependence upon initial conditions will, then, lead to a divergence
between the corresponding anihilation operators' eigenvalues, for
each element of the ensemble.

For values of $K$ sufficiently high to destroy some of the KAM tori,
but not high enough, the dynamics will combine elements of regularity
with randomness, becoming an exemplary picture of complexity in its
interweaving of regularity and randomness (Fig.3(c), Fig.4(left)).

We can find, in the quantum register's dynamics, a quantum computational
example of the three phases of complex systems' distributed computing
dynamics, previously identified within classical approaches to evolutionary
computation: a regular phase (Fig.3(a)), a chaotic phase (Fig.3(b)),
and an intermediate phase (Fig.3(c), Fig.4(left)).

The chaotic phase takes hold after the $K>K_{c}$, when the last KAM
invariant torus has been destroyed, thus, for $K=7$ (Fig.3(b)), we
find the system in a fully chaotic phase, no regular dynamical structures
present. For low $K$ the system shows a predominantly regular dynamical
motion (Fig.3(a)). For $K$ above or below $K_{c}$, but neither too
below nor too above, we find the system in the intermediate phase
with regular and chaotic dynamics (Fig.3(b), Fig.4(left)).

The degree of dominance of chaotic elements in the quantum ensemble
versus the regular elements is dependent upon the value of $K$, this
intermediate phase can be called \emph{complex quantum stochastic
phase}, in classical complex systems science it has been called \emph{edge
of chaos} \cite{key-5,key-11,key-12}.

In the current case, the \emph{complex quantum stochastic phase} allows
one to link the notion of \emph{edge of chaos}, found in classical
computing examples of cellular automata \cite{key-5,key-12} and random
boolean networks \cite{key-11}, to the quantum register dynamics'
phase transition to chaos, where the critical line $K_{c}$ plays
the role of a reference marker around which we find the intermediate
stages of the \emph{complex quantum stochastic phase}.

An example of the \emph{complex quantum stochastic phase} is shown
in Fig.4(left), for $K$ above but very close to $K_{c}$, for an
ensemble of 120 elements. The consequences of this phase for the dynamics
of the quantum averages of each of the $\hat{N}_{j}$ number operators,
plotted as a function of the corresponding phases $\phi_{j}$, for
each register, is shown in Fig.4(right). Thus, the quantum computing,
by making emerge a classical dynamics, leads to a pattern of regular
dynamics, intermixed with randomness at the level of the quantum average
occupation numbers of each register, leading to a basic example of
complex quantum field dynamics (if the \emph{register indices} are
made to correspond to different quantum \emph{field modes}) emerging
from parallel path-dependent quantum computation. Such a complex quantum
field dynamics can also be seen for the dynamics of the probabilities,
for each {}``mode'', of finding zero particles in that {}``mode'',
that is, the quantum probabilties $P_{s}\left[n\left(j\right)=0\right]$
dynamics for each quantum register, as shown in Fig.5 for $K=0.979$.

A striking feature, observable in Fig.5, is the fact that the quantum
computation dynamics drives some of the registers (\emph{field modes})
to chaotic dynamics that very nearly approach zero quanta states$\left|...,n(j)=0,...\right\rangle $
with probabilities close to $1$, which is consistent with an almost
{}``spontaneous collapse'' of the corresponding register's quantum
bosonic oscillator coherent quantum state to one of the Hilbert space
basis eigenstates.

\section{Conclusion}

Classical chaotic dynamics is source of stochastic behavior, emergent
from the nonlinear system's deterministic dynamics, so that, even
in the absence of environmental noise, the system behaves randomly,
demanding a statistical description \cite{key-16}. In quantum systems,
chaos leads to an even greater diversity of dynamical behaviors, for
instance, as addressed in \cite{key-2}, dissipative quantum chaos
in a quantum system can lead to delocalization of wave packets induced
by the instability of chaotic dynamics as well as to localization
due to dissipation, the transition from localization to delocalization
taking place when the dissipation time $\frac{1}{\gamma}$ ($\gamma$
being the dissipation rate) becomes larger than the Ehrenfest time
$t_{E}\sim\frac{\left|\ln\hbar\right|}{\lambda}$, a similar {}``localization/collapse''
phenomenon also takes place with respect to the ground-state energy
eigenstates of bosonic oscillators, in the present article's example,
on the other hand, as addressed in \cite{key-3}, quantum dynamics
can also suppresses chaos by conserving invariant tori that would
otherwise be destroyed.

In the current work, a source of chaotic dynamics in quantum systems
is addressed in terms of quantum computing structures. Sequences of
quantum gates can operate upon an initial input state producing a
sequence of quantum states in which both quantum probabilities as
well as quantum averages fluctuate randomly and show dynamical instability
as an emergent result from the system's quantum computation, thus,
unitary evolution, taking place in a path-dependent iterative way,
can lead to chaotic orbits at the quantum state level, affecting both
probabilities as well as quantum averages.

Dynamical stochasticity in quantum computing systems is, therefore,
possible leading to dynamical instability of eigenvalues with the
usual sensitive dependence upon initial conditions characterizing
chaotic dynamics, to this is added the fact that quantum averages
and quantum probabilities can become mathematically expressed as nonlinear
functions of chaotic maps.

On the other hand, once we pass from the orbit-based picture to the
ensemble picture, the Poincaré maps can become visual representations
of the quantum ensemble's state transition, leading us directly to
a complex quantum systems' approach. The field-like behavior of parallel
quantum computation allows one to theoretically show that a quantum
field may make emerge complex dynamics already exemplifying at that
which is the most fundamental level of nature the three dynamical
phases addressed in complex systems science: the regular, the chaotic
and the complex in-between regime that intermixes randomness and regularity.

For complex quantum systems with chaotic dynamics, this intermixed
regime, takes place as an intermixture between chaotic dynamics in
the quantum states' orbits and regular dynamics, that is, a complex
intermixture between quantum dynamical stochasticity and regularity,
which explains the notion of \emph{complex quantum stochastic phase}.

These results, generalized by the quantum computing formalism, may
help quantum econophysics address major puzzles, for both economics
and finance, such as: the emergence of continuous state-like dynamics
(including chaotic dynamics) in discrete (quantized) variables%
\footnote{Quantities sold, shares transactioned, all of these are discrete,
such as their quoted prices, however, they also show evidence of emergent
hypercomputation \cite{key-8}. Emergent hypercomputation from a discrete
basis is a characteristic feature present in bosonic \emph{qunat}-based
quantum computation, as proven in the present work.%
} \cite{key-8}. Quantum chaos theory is, thus, a fruitful basis for
expanding both complex quantum systems science and quantum econophysics
within the growing framework of a \emph{quantum complex systems science}.

\section*{Figures}

\begin{figure}[H]
\includegraphics[scale=0.35]{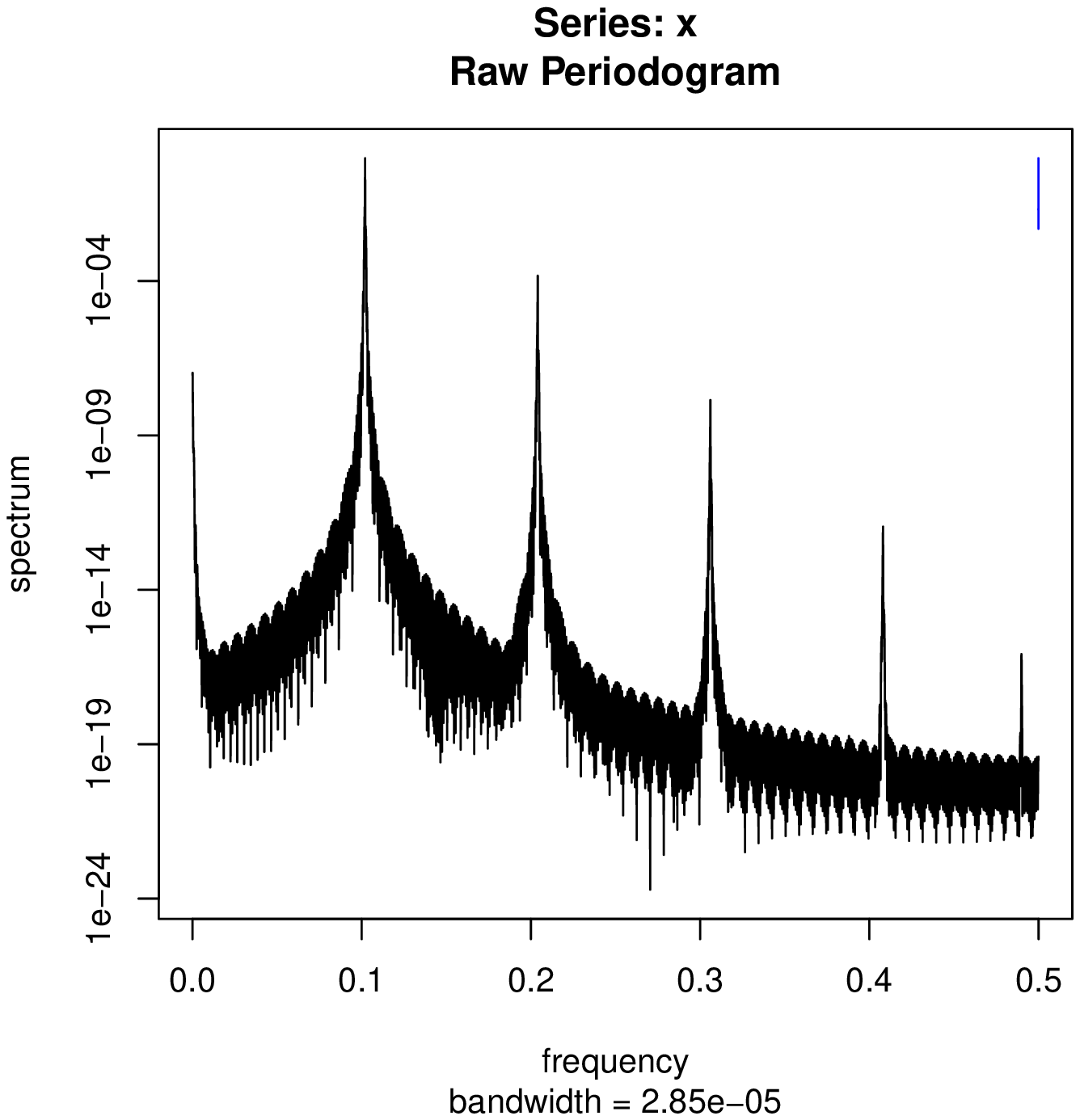}\includegraphics[scale=0.35]{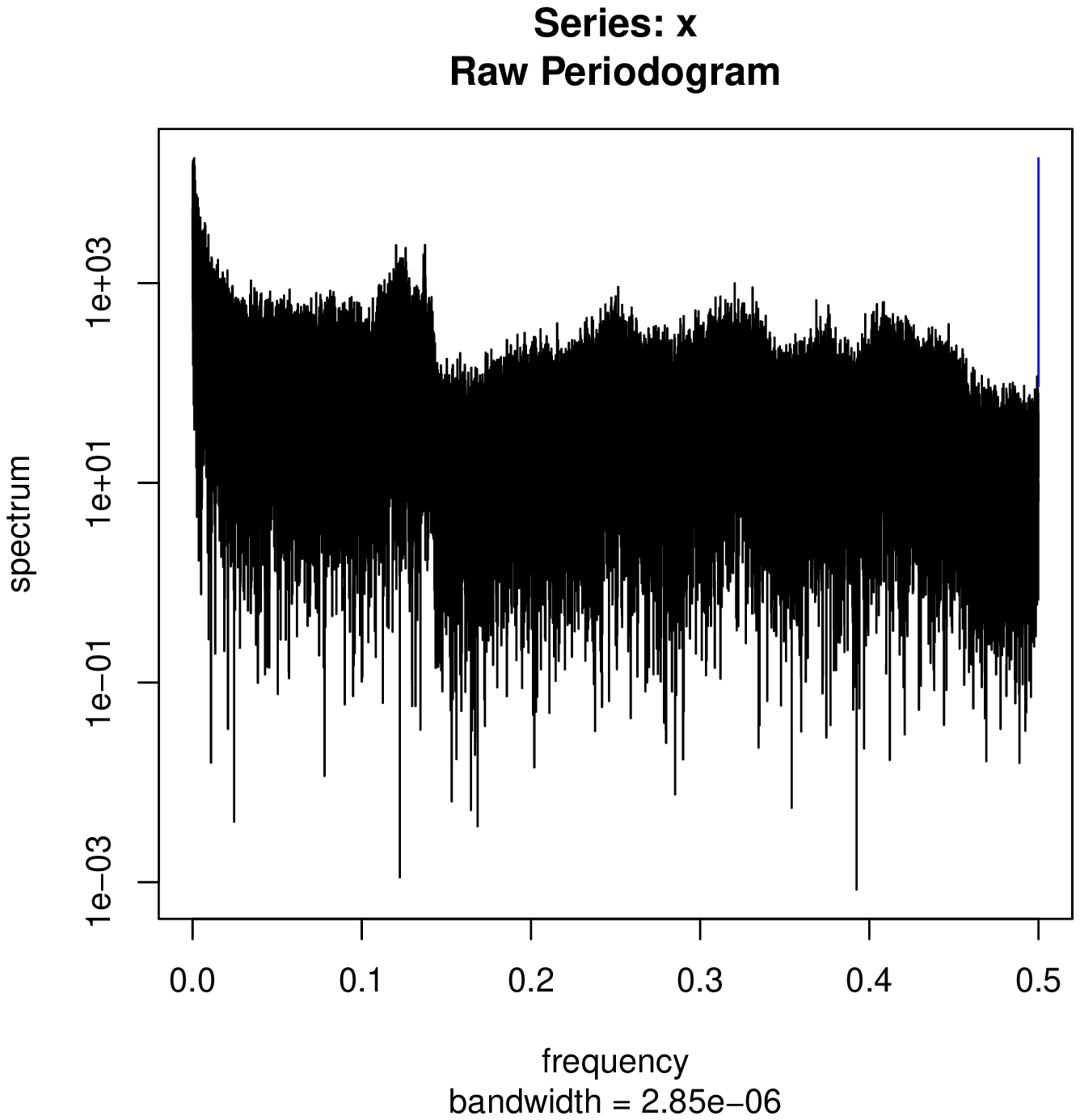}

\caption{{\small Spectrum estimates for $\left\langle \hat{N}\right\rangle _{s}$,
from a Netlogo simulation of the system with $\left|c_{0}\right\rangle =\left|\frac{2\pi}{10}e^{i\frac{3\pi}{10}}\right\rangle $
and the following parameters: $K=0.01$ (left diagram), $K=1.4$ (right
diagram). In both cases 10,000 of the first steps of the simulation
were dropped, with 10,000 additional data points recorded for the
case of $K=0.01$ (which was sufficient to capture the pattern), while
100,000 additional data points were recorded for the case of $K=1.4$
in order to capture more of the pattern.}}
\end{figure}

\begin{figure}[H]
\begin{centering}
\includegraphics[scale=0.35]{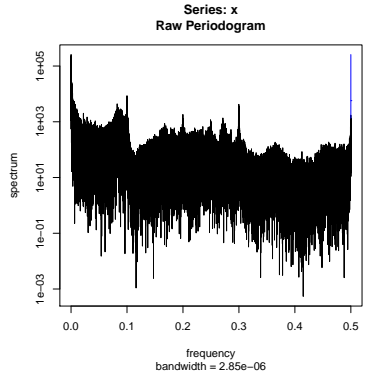}
\par\end{centering}

\caption{{\small Spectrum estimates for $\left\langle \hat{N}\right\rangle _{s}$,
from a Netlogo simulation of the system with $K=0.971635406$. The
initial coherent state was set to $\left|c_{0}\right\rangle =\left|\frac{2\pi}{10}e^{i\frac{3\pi}{10}}\right\rangle $,
10,000 of the first steps of the simulation were dropped, with 100,000
additional data points recorded.}}
\end{figure}

\begin{figure}[H]
\begin{centering}
\includegraphics[scale=0.4]{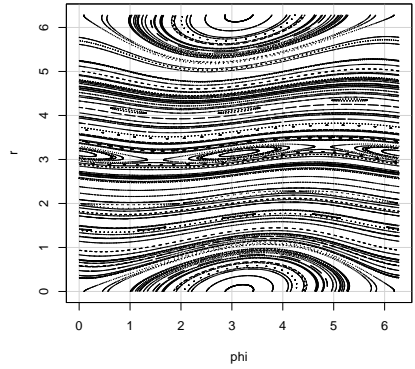}
\par\end{centering}

\begin{centering}
{\small (a)}
\par\end{centering}{\small \par}

\begin{centering}
\includegraphics[scale=0.4]{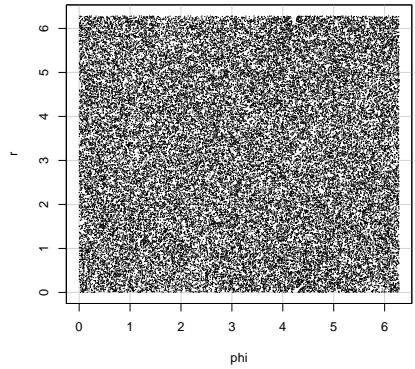}
\par\end{centering}

\begin{centering}
{\small (b)}
\par\end{centering}{\small \par}

\begin{centering}
\includegraphics[scale=0.4]{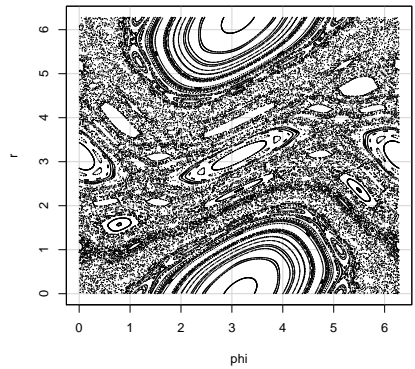}
\par\end{centering}

\begin{centering}
{\small (c)}
\par\end{centering}{\small \par}

\caption{{\small Quantum register machines' dynamics with $m=120$ registers
representing a coherent state ensemble simulated in Netlogo, with
400 iteration steps recorded after 10,000 initial steps, and: (a)
$K=0.3$; (b) $K=7$; (c) $K=1.1$.}}

\end{figure}

\begin{figure}[H]
\includegraphics[scale=0.4]{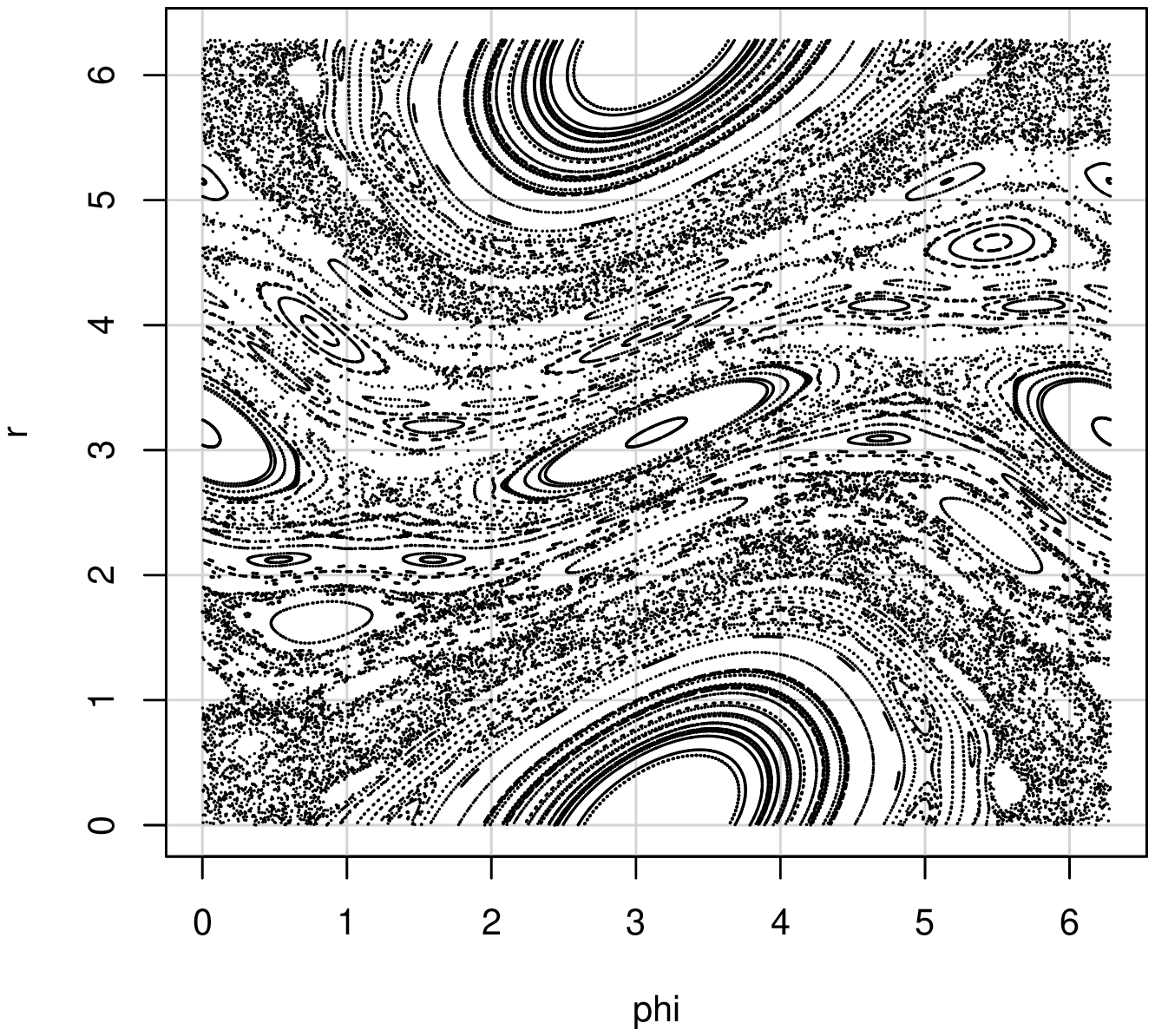}\includegraphics[scale=0.4]{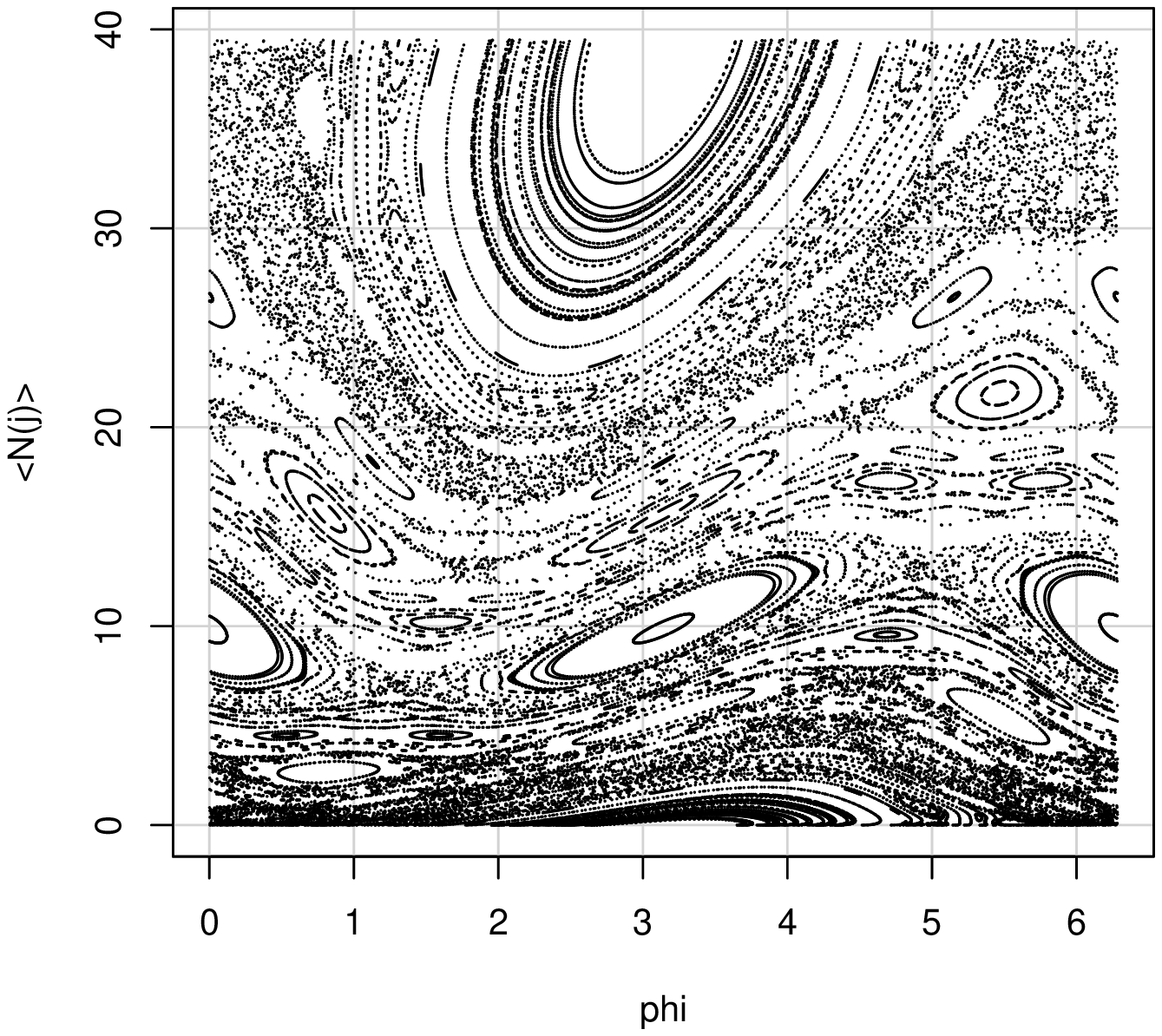}

\caption{\emph{\small Complex quantum stochastic phase}{\small , for $K=0.97163541$,
simulated in Netlogo, for an ensemble of size $m=120$, with 400 iteration
steps recorded after 10,000 initial steps. On the left is the ensemble
dynamics resulting for the quantum register machines computation,
on the right is the corresponding number operator quantum average
dynamics for the same simulation.}}
\end{figure}

\begin{figure}[H]
\begin{centering}
\includegraphics[scale=0.3]{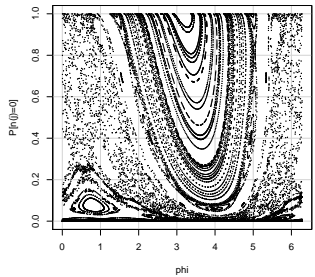}
\par\end{centering}

\caption{Ensemble ground-state quantum probabilities' dynamics, for $K=0.979$,
simulated in Netlogo, for $m=120$ quantum registers, {\small with
400 iteration steps recorded after 10,000 initial steps.}}

\end{figure}

\end{document}